\documentclass[sigconf]{acmart}
\AtBeginDocument{%
  }

\AtBeginDocument{%
  \bibliographystyle{unsrt}         % 按引用顺序编号
}

\makeatletter
\newcommand{\equalfootnote}[1]{%
  \begingroup
  \renewcommand{\thefootnote}{}%
  \long\def\@makefntext##1{%
    \noindent$*$\ ##1%
  }%
  \footnotetext{#1}%
  \addtocounter{footnote}{-1}%
  \endgroup
}
\makeatother

\copyrightyear{2026}
\acmYear{2026}
\setcopyright{cc}
\setcctype{by}
\acmConference[DAC '26]{63rd ACM/IEEE Design Automation Conference}{July 26--29, 2026}{Long Beach, CA, USA}
\acmBooktitle{63rd ACM/IEEE Design Automation Conference (DAC '26), July 26--29, 2026, Long Beach, CA, USA}
\acmDOI{10.1145/3770743.3806375}
\acmISBN{979-8-4007-2254-7/2026/07}

\begin{document}

\title{Scheduling Cause-Effect Chains without Timing Anomalies in End-to-End Latency}

% \author{Yixuan Zhu$^{\dagger *}$, Bo Zhang$^{\dagger *}$, Yinkang Gao$^{\dagger}$, Haoyuan Ren$^{\dagger}$, Cheng Tang$^{\dagger}$, Caixu Zhao$^{\dagger}$, Lei Gong$^{\dagger}$, Teng Wang$^{\ddagger}$, Wenqi Lou$^{\ddagger}$, Xi Li$^{\dagger}$}
% \affiliation{%
%   \institution{$^{\dagger}$University of Science and Technology of China, Hefei 230026,
% China}
%   \institution{$^{\ddagger}$Suzhou Institute for Advanced Research, University of Science
% and Technology of China, Suzhou 215123, China
% }
% }
% \email {{zhuyixuan, sazb, gaoyinkang, jackren, tangcheng, cathyz}@mail.ustc.edu.cn}
% \email{{leigong0203, wangt63, louwenqi,llxx}@ustc.edu.cn}
% \renewcommand{\authors}{Yixuan Zhu, Bo Zhang, Yinkang Gao, Haoyuan Ren, Cheng Tang, Caixu Zhao, Lei Gong, Teng Wang, Wenqi Lou, Xi Li}

\author{Yixuan Zhu*}
\affiliation{%
  \institution{University of Science and Technology of China}
  \city{Hefei}
  \postcode{230026}
  \country{China}
}
\email{zhuyixuan@mail.ustc.edu.cn}

\author{Bo Zhang*}
\affiliation{%
  \institution{University of Science and Technology of China}
  \city{Hefei}
  \postcode{230026}
  \country{China}
}
\email{sazb@mail.ustc.edu.cn}

\author{Yinkang Gao}
\affiliation{%
  \institution{University of Science and Technology of China}
  \city{Hefei}
  \postcode{230026}
  \country{China}
}
\email{gaoyinkang@mail.ustc.edu.cn}

\author{Haoyuan Ren}
\affiliation{%
  \institution{University of Science and Technology of China}
  \city{Hefei}
  \postcode{230026}
  \country{China}
}
\email{jackren@mail.ustc.edu.cn}

\author{Cheng Tang}
\affiliation{%
  \institution{University of Science and Technology of China}
  \city{Hefei}
  \postcode{230026}
  \country{China}
}
\email{tangcheng@mail.ustc.edu.cn}

\author{Caixu Zhao}
\affiliation{%
  \institution{University of Science and Technology of China}
  \city{Hefei}
  \postcode{230026}
  \country{China}
}
\email{cathyz@mail.ustc.edu.cn}

\author{Lei Gong}
\affiliation{%
  \institution{University of Science and Technology of China}
  \city{Hefei}
  \postcode{230026}
  \country{China}
}
\email{leigong0203@ustc.edu.cn}

\author{Teng Wang}
\affiliation{%
  \institution{Suzhou Institute for Advanced Research, University of Science and Technology of China}
  \city{Suzhou}
  \postcode{215123}
  \country{China}
}
\email{wangt63@ustc.edu.cn}

\author{Wenqi Lou}
\affiliation{%
  \institution{Suzhou Institute for Advanced Research, University of Science and Technology of China}
  \city{Suzhou}
  \postcode{215123}
  \country{China}
}
\email{louwenqi@ustc.edu.cn}

\author{Xi Li}
\affiliation{%
  \institution{University of Science and Technology of China}
  \city{Hefei}
  \postcode{230026}
  \country{China}
}
\email{llxx@ustc.edu.cn}

\renewcommand{\shortauthors}{Zhu et al.}

\renewcommand{\shortauthors}{Zhu et al.}
% \renewcommand{\shortauthors}{Zhu et al.}

\begin{abstract}
In real-time systems, both individual task execution and data propagation must meet strict timing constraints. Cause–effect (CE) chains are widely used to analyze such behaviors by end-to-end latency. However, timing anomalies (TAs) can distort it, where a local reduction in execution times
leads to an increase in the overall end-to-end latency.
As a result, precisely analyzing the upper bounds of the latency becomes challenging, and such systems typically exhibit larger upper bounds than TA-eliminated systems.
Existing studies either eliminate TAs by completely sacrificing average latency to simplify analysis or, despite adopting complex safe analysis methods, do not eliminate TAs effectively, still having high latencies.

To address this issue, we identify two basic causes of TAs in end-to-end latency. Based on these causes, we propose the first treatment that eliminates TAs in the latency with negligible average latency loss using Deterministic Data Flow (DDF). We further formally prove its TA-free property. Therefore, we can get a precise upper bound for latency when all jobs execute with their worst-case execution times. 
Experimental results show that it effectively reduces the maximum end-to-end latency, the average latency, and latency jitter compared with the state-of-the-art (SOTA) method.
\end{abstract}

%    <concept>
%        <concept_id>10010520.10010570.10010571</concept_id>
%        <concept_desc>Computer systems organization~Real-time operating systems</concept_desc>
%        <concept_significance>500</concept_significance>
%        </concept>
%  </ccs2012>

% \ccsdesc[500]{Computer systems organization~Real-time operating systems}




% \{llxx,  leigong0203\}@ustc.edu.cn}
\maketitle
\equalfootnote{Equal contribution. Corresponding authors Xi Li and Lei Gong.}
\setcounter{footnote}{0}

\vspace{-0.2cm}
\section{Introduction}
% Real-time systems typically operate as multi-stage processing pipelines. For instance, autonomous driving systems include perception, planning, and control stages. These systems are highly timing-sensitive, requiring control actions to meet strict deadlines to prevent catastrophic outcomes. To model and analyze such timing behaviors, cause–effect (CE) chains are widely used.

% Real-time systems usually operate as multi-stage processing pipe-lines. For instance, autonomous driving systems involve perception, planning, and control stages. These systems are highly timing-sensitive, requiring control actions to be executed within strict deadlines to avoid catastrophic outcomes. To model and analyze such timing behaviors and ensure real-time guarantees, cause–effect (CE) chains are widely adopted.

Real-time systems typically operate as multi-stage processing pipe-lines~\cite{tang2023reaction, jiang2023analysis,tsi}.
For example, autonomous driving systems include perception, planning, and control stages~\cite{li2024data}. These systems are highly timing-sensitive,
requiring control actions to meet strict deadlines to prevent catastrophic outcomes. To model and analyze such timing behaviors,
cause–effect (CE) chains are widely used.
\begin{figure}
    \centering
    \includegraphics[width=\linewidth]{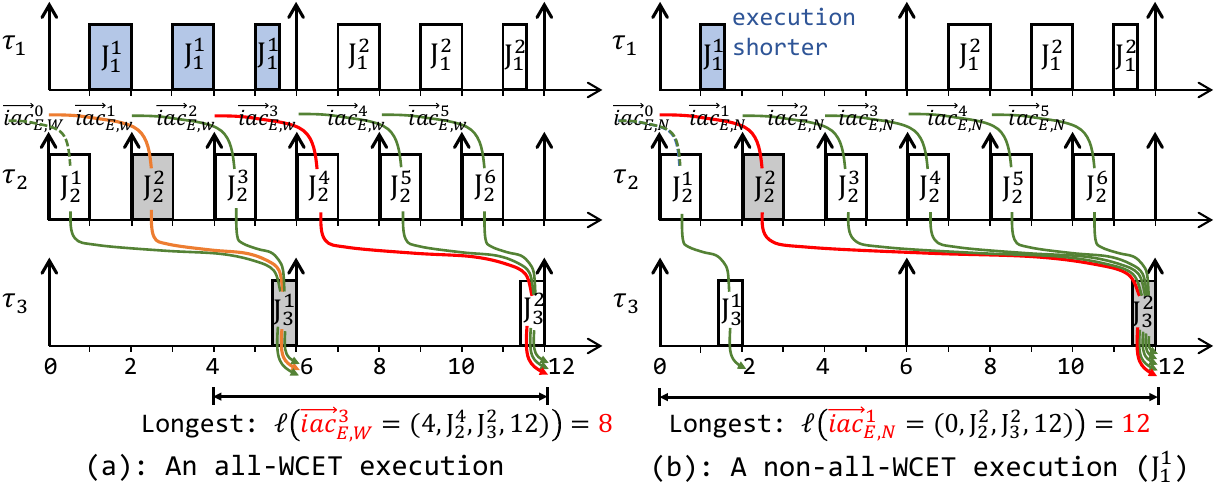}
    % \vspace{-0.8cm}
    \caption{
    Timing anomaly in the data propagation of CE chain $E =\tau_2 \rightarrow \tau_3$ under interference from an unrelated task $\tau_1$.
    The three periodic tasks: $\tau_1 = (6, 0, 0.5, 2.5)$, $\tau_2 = (2, 0, 0.5, 1)$, $\tau_3 = (6, 0, 0.5, 0.5)$ are scheduled by rate monotonic, where each tuple denotes $(period, phase, BCET, WCET)$.
    The colored data path represents the first response to an external event.
    % 在图下方标注了the structure of the longest (red) chain leading to Max Response Time (MRT), denoted as $(externalEventTime, processingJobsChain, actuationTime)$.
    % The colored data path shows how the system first reacts to an external event through a sequence of jobs.
    The structure of the longest (red) chain, which determines the MRT, is annotated below each case in the form $\overrightarrow{iac} = (extEventTime,\ procJobChain,\ actTime)$.
    In case (a), where all jobs are executed by WCETs, the longest (red) path leads to MRT=8.
    In case (b), $J_1^1$ executing shorter results in a longer path, increasing the MRT to 12, which means TA occurs.
    % In case (a), where all jobs execute at their WCETs, the longest (red) path yields an MRT of 8. In case (b), the shorter execution of $J_1^1$ results in a longer path, increasing the MRT to 12.
    % \\
    % In (a), all jobs are executed by WCETs. Multiple immediate forward augmented chains (job-level) are constructed by~\cite{gunzel2023compositional} and get the MRT by the longest, i.e., $\overrightarrow{iac}^3_{E, W}$ with length 8. In (b), $J_1^1$'s time is reduced, not all WCETs. But we can get a longer chain $\overrightarrow{iac}^1_{E, N}$ with length 12, a bigger reaction time than MRT.   
    }
    \label{fig:TAExample}
    \vspace{-0.45cm}
\end{figure}

% A CE chain can be modeled as a linear directed acyclic graph, where nodes represent tasks and edges denote their data dependencies. For example $\tau_2 \rightarrow \tau_3$ indicates that task \(\tau_3\) is data-dependent on \(\tau_2\).
% Its timing behavior is characterized by the end-to-end latency, comprising two key attributes: 
% the Maximum Reaction Time (MRT)—the longest interval between an external event and the first corresponding system response—
% and the Maximum Data Age (MDA)—the longest duration between data acquisition and its actual use, representing data freshness~\cite{guenzel2024end}.
A CE chain can be modeled as a linear directed acyclic graph, where nodes represent tasks and edges denote their data dependencies. For example $\tau_2 \rightarrow \tau_3$ indicates that task \(\tau_3\) is data-dependent on \(\tau_2\).
Its timing behavior is characterized by the end-to-end latency~\cite{davare2007period, vincentelli2007optimizing, zhu2013optimization, hamann2017communication, teper2022end,gunzel2021timing,gunzel2023compositional, TA_analysis_1}, including: 
\emph{Maximum Reaction Time} (MRT)—the longest interval between an external event and the first corresponding system response—
and \emph{Maximum Data Age} (MDA)—the longest duration between data acquisition and its actual use, i.e., data freshness~\cite{guenzel2024end}.

However, even in \emph{time-triggered fixed-priority preemptive} (TT-FP) systems under implicit communication\footnote{Implicit communication means each job reads at its start and writes at its finish.}~\cite{gunzel2021timing,tti}, \textit{timing anomalies} (TAs) can adversely affect the end-to-end latency of CE chains. Taking MRT as an example, TA is a phenomenon where reducing the execution time of tasks on other chains unexpectedly increases the reaction time of the current chain, as shown in Fig.~\ref{fig:TAExample}. 
This phenomenon is reflected in the difficulty of precisely analyzing latency upper bounds and typically larger the upper bounds compared with TA-eliminated systems\footnote{Eliminating TAs has been shown always to reduce worst-case performance metrics of real-time systems, as demonstrated in prior  studies~\cite{TAFree1,dai2021timing,lin2023scheduling,gunzel2021timing}.}. 
This is due to even slight variations in execution times may reshape job-level data propagation paths, leading to larger end-to-end latency than all-WCET execution and even period-level increases. 
For instance, in Fig.~\ref{fig:TAExample}, under the all-WCET case (a), job $J_2^2$ first influences $J_3^1$, whereas under the non-all-WCET case (b), it first influences $J_3^2$, resulting in a reaction time increase.
% increased by task $\tau_3$'s period.
% Moreover, for end-to-end analysis in systems with TAs, exhaustive timing scenario analysis is infeasible; only an over-approximation of the worst case, which may never occur in practice, can be constructed.
Moreover, for end-to-end analysis in systems with TAs, exhaustive timing scenario analysis is infeasible; only an approximate worst-case construction, which may never occur, can be derived.
Therefore, by constraining possible data-propagation paths, it may eliminate pessimistic scenarios in TA systems. 

Existing studies~\cite{graham1969bounds,voudouris2017timing,TAFree1,dai2021timing, lin2023scheduling,lin2025quasi} about TAs mainly focus on the worst-case response time~\cite{maiza2019survey}. While research on TAs in end-to-end latency of CE chains is limited, it can be categorized into two directions. 
First, work~\cite{gunzel2021timing} actively eliminates TAs to simplify analysis at the scheduling level by enforces all tasks to execute online with WCETs. Thereby, obtaining safe upper bounds directly from all-WCET execution offline. However, since task execution times are usually shorter than WCETs, this treatment results in a large average latency.
Second, other works~\cite{gunzel2023compositional, TA_analysis_1} propose safe analysis methods that do not modify the scheduling, yielding safe latency upper bounds that account for TA effects. 
For example,~\cite{gunzel2023compositional} proposes an analysis framework that combines simulated execution with abstract chains to get safe bounds without sacrificing average latency. 
Nevertheless, since TA persists in these systems, the upper bounds on latency may remain high. 
This limitation may pose drawbacks for real-time systems that rely on worst-case behavior.

% To address these problems, we propose a treatment based on Deterministic Data Flow (DDF)  to eliminate TA in end-to-end latency under FP scheduling, without sacrificing average performance. 
% To address these problems, we propose the first treatment that eliminates TA in end-to-end latency with negligible loss in average latency by enforcing Deterministic Data Flow (DDF). DDF is a set of fixed communication relations between jobs. We realize DDF by achieving two objectives: 1) \textbf{RAW}  --- enforcing deterministic job-level producer–consumer execution dependencies while maintaining fixed-priority (FP) scheduling, and 2) \textbf{RFR} --- ensuring that each job \textbf{reads} data \textbf{from} its \textbf{right} producer as specified in the DDF semantics. Specifically, in the offline, we perform a WCET-based simulation under the original FP scheduling to extract job-level communication relations as the DDF and derive execution dependencies between jobs from the DDF. 
% We further modify task attributes to ensure that the derived execution dependencies are preserved during online FP scheduling.

To address these problems, we propose the first treatment that eliminates TA in end-to-end latency with negligible average latency loss by enforcing Deterministic Data Flow (DDF). DDF is a set of fixed communication relations between jobs. We realize it by achieving two objectives: 
1) \textbf{ Read–After–Write (RAW)}  --- 
enforcing deterministic job-level producer–consumer execution dependencies, while maintaining TT-FP scheduling, 
% enforcing the start of a consumer job's read only after the finish of its producer job's write
and 2) \textbf{Read–From–Intended (RFI)} --- ensuring that each job reads data from its intended producer as specified in the DDF. 
Specifically, in the offline, we perform an all-WCET execution under the original FP scheduling to extract job-level communication relations as a DDF and derive job-level execution dependencies from the DDF. 
%the DDF and derive execution dependencies between jobs from the DDF. 
We further modify task attributes to ensure that the derived execution dependencies are preserved during online TT-FP scheduling.
% To ensure the analyzability and provable safety of the latency bound, we adjust task attributes offline to integrate the execution dependencies into a unified new FP scheduling for online execution. 
During online execution,
for DDF objective 1) \textbf{RAW}, the \textbf{modified task attributes} ensure that job-level producer–consumer execution dependencies are preserved;
for DDF objective 2) \textbf{RFI}, a \textbf{multi-buffer communication mechanism}~\cite{caspi2008semantics} guarantees each job correctly reads data from its intended producer as defined in the DDF. Then, we prove its TA-free attribute by the invariance of job-level dataflow relations and the safe boundaries of jobs' read and write times under our treatment. 
For the latency analysis, due to TA-free, we can directly get the precise maximum latency via an all-WCET execution. 
Experiments show it can reduce MRT by 9\%–12\% on average, reduce average latency by 11\%-41\% compared to other TA-free treatment, and lower latency jitter by 53\%–68\%. 
Our contributions are as follows: 
% (up to 61\% in the best case)
% \vspace{-0.45cm}
% Experimental results show that it can reduce MRT by 7\%–11\% on average (up to 64\% in the best case), improving average performance by 10\%, and lowering execution jitter by 14\%–22\%. Our contributions are as follows: 
\begin{itemize}
    \item \textbf{Key Insight}. We identify two causes of TA in end-to-end latency: variations in structural change to job-level dependency (immediate forward) chains and job read/write times.
    \item \textbf{TA-free Treatment}. To the best of our knowledge, we present the first treatment to eliminate TA in end-to-end latency for cause–effect chains under FP scheduling with negligible average end-to-end latency loss.
    \item \textbf{Modeling Approach}. We define a more general immediate forward chain, enabling end-to-end latency analysis 
    for both register and multi-buffer communication.
    % for more communication models.
    \item \textbf{Correctness Proof of our treatment}. We formally prove our treatment is free from TAs in end-to-end latency. 
    \item \textbf{Experimental Evaluation}. Results show our method can effectively reduce the maximum end-to-end latency, average end-to-end latency, and end-to-end latency jitter.  
\end{itemize}

% \vspace{-0.5cm}
\section{PRELIMINARY}
\subsection{System Model}
\label{lab:System Model}
In this paper, cause–effect (CE) chains and CE graphs are used to model.
The CE chain is represented as a linear directed acyclic graph (DAG), denoted by
$E = {\tau_1 \rightarrow \tau_2 \rightarrow \cdots \rightarrow \tau_n}$,
where each node is a periodically time-triggered (TT) task, and the "$\rightarrow$" indicates that the succeeding task data depends on the preceding one.
Each task $\tau_i$ is characterized by a quadruple $(T(\tau_i), O(\tau_i),  B(\tau_i), W(\tau_i))$, denoting its period, phase, BCET, and WCET.
The actual execution time of a task instance lies within the interval $[BCET, WCET]$.
Multiple CE chains that share common nodes can be combined into a CE graph.
The system may consist of multiple disjoint CE graphs that collectively describe the periodic tasks set $\mathbb{T}$.

Tasks set $\mathbb{T}$ execute on a single-core platform under an FP scheduling policy, such as the RM scheduling algorithm.
Each task can generate multiple execution instances, referred to as jobs.
The $k$-th job of task $\tau_i$ is denoted by $J_i^k$, with release time $r(J_i^k)$, start execution time $s(J_i^k)$, and finish execution time $f(J_i^k)$.
% Since the system adopts an implicit communication model~\cite{gunzel2023compositional}, the read event time of a job $J_i^k$, denoted as $re(J_i^k)$, is equal to $s(J_i^k)$; The write event time $we(J_i^k)$ is equal to $f(J_i^k)$.

Since the system uses an implicit communication~\cite{gunzel2023compositional},
the read and write event times of job $J_i^k$ are aligned with the start and finish of it. Therefore, the read and write event time $re(J_i^k)$ and $we(J_i^k)$ of a job, is respectively equal to its $s(J_i^k)$ and $f(J_i^k)$. But the first task of a CE chain typically corresponds to an external sensor sampling task executed at a fixed sampling interval. To avoid sampling jitter that may degrade control performance, it is commonly assumed the jobs of such tasks sample immediately upon release, i.e., $re(J_i^k) = r(J_i^k)$.
\vspace{-0.7cm}

\subsection{Job Chain and End-to-End latency}
\label{sec:Job Chain and End-to-End latency}
% The job chains are the foundation for analyzing the MRT and MDA.
% Here, we adopt the definition from Dürr et al.~\cite{durr2019end} to describe it. The propagation of data along the CE chain $E$ can be represented by a job chain, which is a sequence of jobs $(J_1, J_2, \dots, J_{|E|})$. Each $J_i$ corresponds to the job of the $i$-th task in $E$.
% These jobs satisfy the data dependency: for any $i \in \{1, \dots, |E|-1\}$, $we(J_i) \le re(J_{i+1})$.

To capture the data propagation, Dürr et al.~\cite{durr2019end} introduced the immediate forward job chain $\overrightarrow{ic}_{E, S}^m$, which describes the first response process to the $m+1$-th job’s sampled data for CE chain E in schedule $S$. Specifically, the direct successor job of each job is defined as the first job that responds to its output data. 

However, their definition applies only to register communication and cannot be generalized to more flexible models, such as multi-buffer communication.
To this end, we provide a more general definition of immediate forward chain based on the communication relation (data dependency) between jobs. Therefore we define a function $R_{\text{com}}(J_i)$ represents the set of jobs communicate with $J_i$;

% \textbf{Definition 1} (immediate forward job chain): For any $i \in \{1, 2, $ \\ $\dots, |E|-1\}$, the job $J_{i+1}$ is defined as the first read job in task $E(i+1)$ that is affected by its predecessor job $J_i$. 

\textbf{Definition 1} (immediate forward job chain): 
For the $m+1$-th job $J_1^{m+1}$ of the head task in a CE chain $E$ under schedule $S$, the immediate forward job chain $\overrightarrow{ic}_{E, S}^m$ is defined as a sequence of jobs
$(J_1, J_2, \dots, J_{|E|})$, where $J_1=J_1^{m+1}$.
For any $i \in [ 1, |E|-1 ]$, the $J_{i+1}$ is the first read job of $E(i+1)$ that is affected by its predecessor $J_i$:
% For the $m+1$-th job $J_1^{m+1}$ of head task in CE chain $E$ under scheduling $S$, 
% $\overrightarrow{ic}_{E, S}^m$ is a sequence of jobs $(J_1, J_2, \dots, J_{|E|})$ .
% Where, for any $i \in \{1, 2, \dots, |E|-1\}$, the $J_{i+1}$ is defined as the first read job of task $E(i+1)$ affected by its predecessor $J_i$: %with two cases:

\textbf{Case 1}: If $J_i$ appears as a writing job for $E(m+1)$ in the $R_{\text{com}}$, i.e., $J_i \in wj(R_{\text{com}}, m+1)$, then $J_{i+1}$ is defined as the first reading job that directly communicates with it: $J_{i+1} = arg \ min_{J \in R_{\text{com}}(J_i)} id(J)$. 
% \textbf{Case 1}: If $J_i$ appears as a writing job in the $R_{\text{com}}$, i.e., $J_i \in wj(R_{\text{com}})$, then $J_{i+1}$ is defined as the first reading job that directly communicates with it: $J_{i+1} = arg \ min_{J \in R_{\text{com}}(J_i)} id(J)$. 

\textbf{Case 2}: If $J_i$ does not appear as a writing job for $E(m+1)$, i.e., $ J_i\notin wj(R_{\text{com}}, m+1)$, it indirectly influences subsequent jobs from the same task $E(i)$ through internal state propagation~\cite{zhang2025optimizing}. These subsequent jobs may, in turn, affect later reading jobs by Case 1. Formally, $J_{i+1} = arg \ min_{J \in \{R_{\text{com}}(J_i')|J_i' \in sub(J_i) \land J_i' \in wj(R_{\text{com}},m+1)\}} id(J)$.

Here, $sub(J_i)$ denotes the set of subsequent jobs of the same task of $J_i$;
and $id(J)$ is the id number of job $J$.

% Furthermore, to capture the complete data-flow behavior, G{\"u}nzel et al.~\cite{gunzel2021timing} define the corresponding immediate forward augmented job chain $\overrightarrow{iac}_{E, S}^m$ on the CE chain $E$ and schedule $S$. Specifically, by adding an external event and an actuation event to the beginning and end of $\overrightarrow{ic}_{E, S}^m$, $\overrightarrow{iac}^m_{E, S}$ is represented as $(z, J_1, J_2, \dots, J_{|E|}, z')$. The external event occurs immediately after the $m$-th sampling, i.e., $z = r_e^{E(1)}(m)$; the first job corresponds to the next instance of the first task (sampling job) in the CE chain $E$, i.e., $J_1 = E(1)(m+1)$; and the actuation corresponds to the write event of the last job, i.e., $z' = w_e(J_{|E|})$. The chain's length is Eqn. (x).

Furthermore, to capture the timing of data-propagation, G{\"u}nzel et al.~\cite{gunzel2021timing} define the immediate forward augmented job chain.
% , with length Eqn.~\ref{lab:lengthComp}.
% Specifically, by adding an external event and an actuation event to the beginning and end of $\overrightarrow{ic}_{E, S}^m$, $\overrightarrow{iac}^m_{E, S}$ is represented as $(z, \overrightarrow{ic}_{E, S}^m, z')$. 

\textbf{Definition 2} (immediate forward augmented job chain): The $\overrightarrow{iac}_{E, S}^m$ is the sequence of $(z, \overrightarrow{ic}_{E, S}^m, z')$, where
\vspace{-0.1cm}
\begin{itemize}
    \item $\overrightarrow{ic}_{E, S}^m=(J_1, \dots, J_{|E|})$ is obtained by Def. 1, where $J_1=J_1^{m+1}$;
    \item 
    $z = re(J_1^m)$ denotes the time of the earliest external activity that could be sampled by the first job $J_1^{m+1}$ of $\overrightarrow{ic}_{E,S}^m$, i.e., the external activity occurring immediately after 
    the job $J_1^{m+1}$.
    \item 
    $z'=we(J_{|E|})$, which is the time of actuation corresponding to the write event of the last job $J_{|E|}$ of $\overrightarrow{ic}_{E,S}^m$.
    % \item 
    % $z = re(J_1^m)$ represents the time of external activity occuring immediately after the $m$-th job of task $E(1)$;
    % \item $\overrightarrow{ic}_{E, S}^m$ is constructed by Definition 1;
    % \item 
    % $z' $ equals $ we(J_{|E|})$, which is the time of actuation corresponding to the write event of the last job in forward chain.
\end{itemize}
The chain $\overrightarrow{iac}_{E, S}^m$'s length is 
\vspace{-0.2cm}
\begin{align}
\ell(\overrightarrow{iac}_{E, S}^m) = z' - z.
\label{lab:lengthComp}
\end{align}

% Furthermore, to capture the timing of data-propagation, G{\"u}nzel et al.~\cite{gunzel2021timing} define the immediate forward augmented job chain $\overrightarrow{iac}_{E, S}^m$. 
% % Specifically, by adding an external event and an actuation event to the beginning and end of $\overrightarrow{ic}_{E, S}^m$, $\overrightarrow{iac}^m_{E, S}$ is represented as $(z, \overrightarrow{ic}_{E, S}^m, z')$. 
% Specifically, it can be represented as $(z, \overrightarrow{ic}_{E, S}^m, z')$ by adding the time of an external activity and an actuation to the beginning and end of $\overrightarrow{ic}$.
% The external activity occurs immediately after the $m$-th sampling, i.e., $z = r_e^{E(1)}(m)$; the first job corresponds to the next instance of the first task (sampling job) in the CE chain $E$, i.e., $J_1 = E(1)(m+1)$; and the actuation corresponds to the write event of the last job, i.e., $z' = w_e(J_{|E|})$. The chain's length is Eqn.~\ref{lab:lengthComp}.
% \vspace{-0.2cm}

Our definition is compatible with register communication, which can be achieved by defining the communication relation $R_{\text{com}}$ such that each reading job always reads from the most recent writing job.
In Fig. ~\ref{fig:TAExample}(b), the $\overrightarrow{iac}^2_{E,S}$ is given by $(0, J_2^2, J_3^2, 12)$ with $l(\overrightarrow{iac}^2_{E, S}) = 12$.

The MRT of CE chain $E$ under scheduling $S$ is the upper bound of the length of all valid $\overrightarrow{iac}$s according to the definition, i.e., 
% The MRT of CE chain $E$ under scheduling $S$ is the upper bound of the lengths of all valid $\overrightarrow{iac}$s derived from a specific job-execution-time combination represented by $S$.
% The MRT of CE chain $E$ under schedule $S$ is the upper bound of the length of all valid $\overrightarrow{iac}$ (Eqn.~\ref{label:CompMRT1}).
%, where a chain is valid if it appears during the system’s steady-state operation~\cite{gunzel2023compositional}.
\vspace{-0.1cm}
\begin{align}
\mathrm{MRT}(E,S) := \max \Bigl\{ \ell\bigl(\overrightarrow{iac}_{E,S}^m\bigr) \ \big|\ m \in \mathbb{N},\ \overrightarrow{iac}_{E,S}^m\ \text{valid} \Bigr\}.
\label{label:CompMRT1} 
\end{align}

% \vspace{-0.5cm}

Nevertheless, since $\mathrm{MRT}$ and $\mathrm{MDA}$ are proven equivalent in Thm.14~\cite{gunzel2023equivalence}, we use $\mathrm{MRT}$ as the representative example later.

\section{Key Insight for Timing Anomalies}
\label{sec:insight}

For the MRT, TAs occur when reducing the execution times of jobs in other chains results in a longer reaction time for the current chain. More precisely, under an all-WCET execution, the longest immediate forward augmented job chain $\overrightarrow{iac}$ obtained may not be the actual longest one due to varying job execution times.
The analysis of MRT can be summarized in two steps:

(1). Construct all immediate forward
job chains $\overrightarrow{ic}$s and $\overrightarrow{iac}$s, which depend on the jobs' read and write event times.

(2). Identification of the longest $\overrightarrow{iac}$, whose length is the MRT.
As defined in Eqn.~\ref{lab:lengthComp}, the length is the duration between $z'$ (the actuation time) and $z$ (the earliest possible time of the external activity that the actuation at $z’$ responds to).
% They correspond to the time of the previous sampling job's read event and the write event of the last job in the augmented job chain.

By the two steps, TAs of MRT arise from the \textbf{two causes}:

(1).
% All-WCET scheduling does 构造出所有的  $\overrightarrow{ic}$s among all possible schedulings.
The $\overrightarrow{ic}$s constructed by all-WCET execution may not cover all possible chains (composed of different jobs) that could arise from all non-all-WCET executions. 
For example, in Fig.~\ref{fig:TAExample}(b) a new additional chain $\overrightarrow{ic}^{2}_{E,S} = ( J_2^2, J_3^2)$ created under non-all-WCET.

% The $\overrightarrow{ic}$s constructed by WCET may not cover all possible chains (composed of different jobs) that could arise from different combinations of job execution times. 
% For example, in Fig.~\ref{fig:TAExample}(b) a new additional chain $\overrightarrow{ic}^{2}_{E,S} = ( J_2^2, J_3^2)$ created under non-WCET.

(2). 
Even if all possible $\overrightarrow{ic}$s are enumerated, for a given $\overrightarrow{ic}$, the $\overrightarrow{iac}$ generated by all-WCET execution is not guaranteed to be the longest one.
% all-WCET scheduling does not guarantee that the $\overrightarrow{iac}$ constructed is the longest one.
% among all possible schedulings.
This is because the read and write event times of jobs corresponding to $z$ and $z’$ of $\overrightarrow{iac}$ may not be guaranteed to bound those that may occur under all non-all-WCET executions.

% (1). The $\overrightarrow{ic}$s constructed by WCET may not cover all possible chains (composed of different jobs) that could arise from different combinations of job execution times. 
% Meanwhile, the WCET-based $\overrightarrow{iac}$s derived from these $\overrightarrow{ic}$s may not include the longest chain. For example, in Fig.~\ref{fig:TAExample}(b) an longer chain $\overrightarrow{iac}^{2}_{E,S} = (0, J_2^2, J_3^2, 12)$ produces under non-WCET.

% (2). The read and write event times
% of each job obtained by WCET-based scheduling are not guaranteed to be the bounds of other schedulings from combinations of various execution times. 
% \section{Treatment of Eliminating Timing Anomalies in End-to-End Latency}

% This chapter addresses the two causes of timing anomalies identified in Section~\ref{sec:insight} and, based on this, proposes a treatment that eliminates TAs in end-to-end latency.
% \vspace{-0.8cm}
\section{Treatment of Eliminating Timing Anomalies}

To address the two causes of timing anomalies identified in Sec.~\ref{sec:insight}, we propose a treatment that eliminates TAs in end-to-end latency.

(1) For Cause 1: to avoid new immediate forward chains in non-all-WCET executions relative to the all-WCET case, we enforce Deterministic Data Flow (DDF), as proved in Thm.~\ref{thm:inva}. DDF fixes the communication relations between producer and consumer jobs through two objectives: \textbf{RAW} (Read--After--Write), requiring a consumer job to start reading only after its producer job finishes writing, and \textbf{RFI} (Read--From--Intended), requiring each consumer job to read from its intended producer job even if earlier outputs from the same producer task are available.
(2). For cause 2: Our goal is to ensure the offline-determined boundaries of jobs' read and write times remain safe in all-WCET and non-all-WCET executions online. This is achieved by realizing DDF without altering the time-triggered FP (TT-FP) scheduling , thereby leveraging its property in Thm.~\ref{lem:boudrw} to fulfill the goal.

% (2). Unifying the scheduling under PF, such that the read and write event times of each job obtained during offline analysis constitute the lower and upper bounds, respectively, of the corresponding event times during online execution. The theoretical foundation of this property will also be presented in the next chapter.

% An overview of our treatment is in Fig.~\ref{fig:overview}, comprising offline analysis and online execution. The blue modules are offline processes, while the green ones are online execution. The analysis of the latency and TA-free proof will be detailed in the next chapter.

% Our treatment is shown in Fig.~\ref{fig:overview}. In the offline, the DDF is constructed by WCET-based execution, and job-level execution dependencies are enforced by modifying task attributes to ensure the RAW. In the online, tasks are scheduled with new attributes under FP, while multi-buffer communications ensure the RFI, preserving the offline DDF and thereby eliminating TAs proved 
% in Theorem~\ref{them:tafree}.

\begin{figure}
    \centering
    \includegraphics[width=\linewidth]{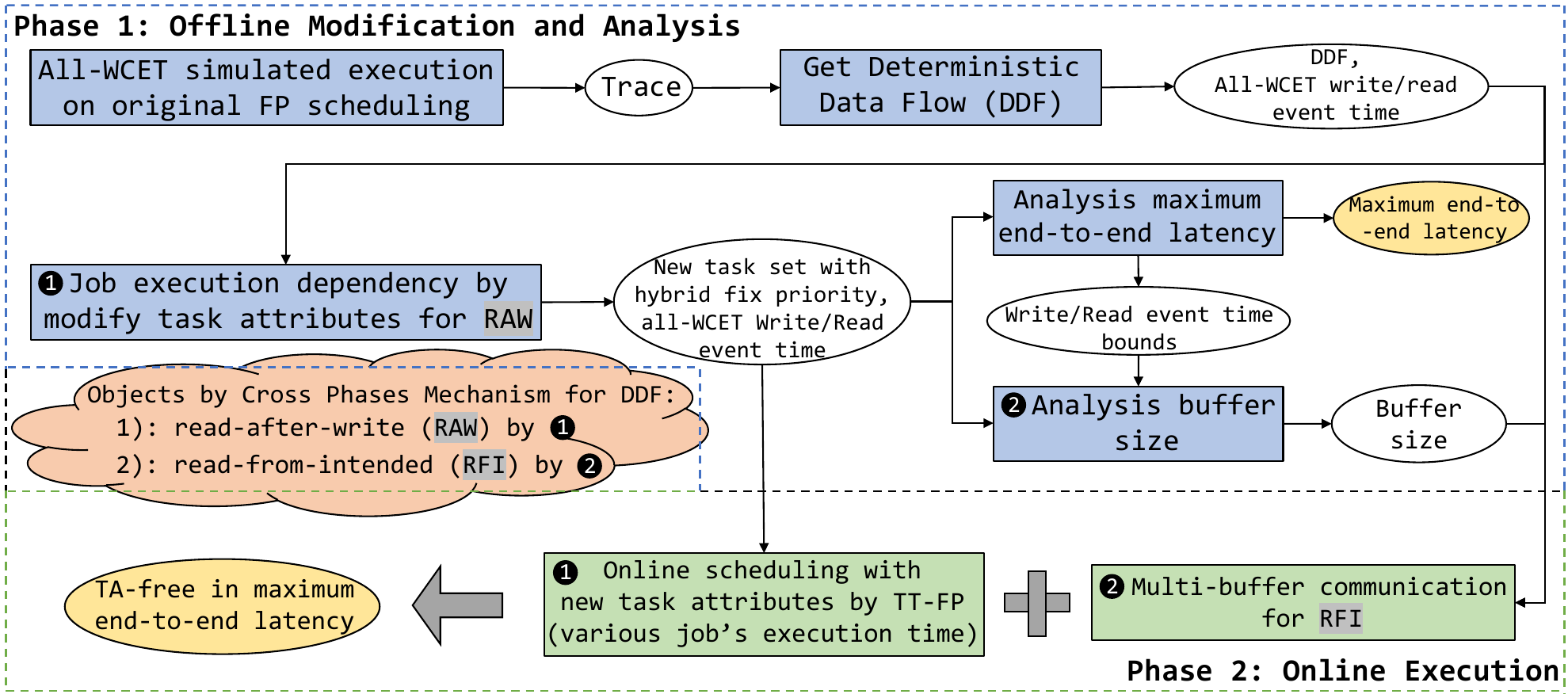}
    \vspace{-0.7cm}
    \caption{Overview of our treatment for TA-free
    in Latency}
    \label{fig:overview}
    \vspace{-0.5cm}
\end{figure}

Our treatment is shown in Fig.~\ref{fig:overview}. In the offline, the DDF is constructed by all-WCET execution, and RAW is ensured by modifying task attributes under TT-FP scheduling.
Online, besides enforcing RAW through scheduling, RFI is ensured by the multi-buffer communication mechanism. Ultimately, RAW and RFI jointly preserve the offline DDF and thereby eliminate TAs, as proven in Thm.~\ref{them:tafree}.

\vspace{-0.15cm}
\subsection{Obtaining Deterministic Data Flow (DDF)}
First, in the offline, we perform an all-WCET simulated execution under the original scheduling to obtain a scheduling trace recording the start and finish times of jobs.
Only a simulation over $[0,2H + O_{\max}]$ is required, where $H$ is the hyperperiod and $O_{\max}$ is the max task phase, as the schedule repeats~\cite{gunzel2023compositional} beyond $2H + O_{\max}$.

Then the DDF is extracted from the schedule trace under register communication, where each read job receives data from the most recent writer. In Fig.~\ref{fig:twoMechainForDDF}(b), the solid arrows denote the job-level communication relations of the DAG in Fig.~\ref{fig:twoMechainForDDF}(a) over one hyperperiod under all-WCET execution. The colored task-level causal relations in the CE chain are thus mapped to the corresponding same-colored job-level communication relations, which together form the DDF.

\vspace{-0.15cm}
\subsection{Realization of DDF under TT-FP Scheduling}
The goal of realizing the DDF is to ensure that, during online execution, inter-job communication remains consistent with the DDF despite variations in jobs' execution times. To achieve this, we introduce two cross-phase mechanisms: modification of the task's attribute under the TT-FP Scheduling 
and a multi-buffer communication mechanism, which respectively ensure the \textbf{RAW} and \textbf{RFI}.
% The former enforces that a read job starts only after its corresponding write job has produced the required data (finish execution), while the latter ensures, when multiple data versions coexist, the read job accesses the correct one—thus preserving the offline read–write behavior in the DDF.

\subsubsection{Modification of Task Attribution for Read-After-Write (RAW)} 
\label{sec:Modification of Task Attribution}
\ 
\newline
\indent
% Realizing the DDF requires that the corresponding producer consumer execution dependencies (RAW) be preserved during online execution. In Fig.~\ref{fig:twoMechainForDDF}(b), the DDF of the task $\tau_3$ and $\tau_6$ is defined by the blue solid arrows, and the resulting execution dependencies are shown in Fig.~\ref{fig:twoMechainForDDF}(c) as the blue dotted arrows. These dependencies must be enforced online.
% Specifically, to preserve the communication relation $J_6^1$ \textit{reads} $J_4^1$ in the DDF of Fig.~\ref{fig:twoMechainForDDF}(b), the start of $J_6^1$ in the online execution of Fig.~\ref{fig:twoMechainForDDF}(c) must be delayed until $J_4^1$ finishes.
Realizing DDF requires preserving the corresponding producer--consumer execution dependencies (RAW) during online execution. In Fig.~\ref{fig:twoMechainForDDF}(b), the blue solid arrows denote the DDF between $\tau_3$ and $\tau_6$, and the resulting execution dependencies are shown as blue dotted arrows in Fig.~\ref{fig:twoMechainForDDF}(c). Specifically, to preserve the communication relation $J_6^1$ \textit{reads} $J_4^1$ in the DDF of Fig.~\ref{fig:twoMechainForDDF}(b), the start of $J_6^1$ in the online execution of Fig.~\ref{fig:twoMechainForDDF}(c) must be delayed until $J_4^1$ finishes.

To enforce the desired execution dependencies, explicit constraints are introduced between the job pairs. A simple solution is to adopt an ET model, e.g., letting the finish of $J_4^1$ trigger the release of $J_6^1$. However, this would yield a hybrid ET-TT system, thereby complicating latency analysis and safety proofs. Instead, we enforce equivalent dependencies by modifying task attributes while preserving TT-FP: each reader is assigned a release time no earlier than its writer’s and an additional strictly lower priority, preventing it from preempting the writer. The details are as follows:

We instantiate each job in the hyperperiod as a new task for a new task set $\mathbb{T'}$, whose execution time interval inherits from its old task, and the period is the hyperperiod of the old task set $\mathbb{T}$.

The initial phase of each new task is  $O(\gamma_i) = T(ort(\gamma_i)) * id(\gamma_i) + O(ort(\gamma_i))$, where $\text{ort}(\gamma_i)$ denotes the original task in $\mathbb{T}$ for new task $\gamma_i$, and $\text{id}(\gamma_i)$ represents the id of the job in the offline phase's all-WCET execution to which $\gamma_i$ corresponds.
If $\gamma_i$ is a reader, let $\mathbb{WS} (\gamma_i)$ denote the set of all writers in DDF for $\gamma_i$. Its phase should be adjusted to ensure that its release time does not precede that of its dependent writer.
Formally, this adjustment is defined as  $O^*(\gamma_i) = max(O(\gamma_i), max_{J' \in \mathbb{WS}(\gamma_i)}O^*(J'))$.

% New task is assigned a priority group, denoted as $\text{NP}$, where: $\text{NP}[0]$ is the priority of its corresponding task in the old task set $\mathbb{T}$; $\text{NP}[1]$ records the execution order among job pairs that communicate in DDF. The writing task precedes the corresponding reading task, as shown in Fig.~\ref{fig:twoMechainForDDF}(c). The priority arbitration mechanism is:
% First determines execution orders according to $\text{NP}[1]$ if tasks have relationships in DDF, selecting the subset ($HPS$) of tasks with the highest precedence.
% Since the execution order is pairwise among communicating tasks, which is a partial order, $HPS$ may contain multiple tasks.
% Then, a secondary arbitration is performed to select the task with the highest priority in the $HPS$ according to $\text{NP}[0]$.

Each new task is assigned a priority group $\text{NP}$, where $\text{NP}[0]$ denotes the priority of its corresponding task in the original task set $\mathbb{T}$, and $\text{NP}[1]$ specifies the DDF-induced execution order between communicating job pairs, with each writing task preceding its corresponding reading task, as shown in Fig.~\ref{fig:twoMechainForDDF}(c). Priority arbitration proceeds in two stages: tasks are first filtered by $\text{NP}[1]$ to obtain the highest-precedence subset $HPS$, and then the task with the highest $\text{NP}[0]$ in $HPS$ is selected for execution. Since the DDF-induced order is partial, $HPS$ may include multiple tasks.

% Each new task is assigned a priority group, denoted by $\text{NP}$. Specifically, $\text{NP}[0]$ denotes the priority of its corresponding task in the original task set $\mathbb{T}$, while $\text{NP}[1]$ captures the execution order of job pairs involved in DDF, where each writing task precedes its corresponding reading task, as illustrated in Fig.~\ref{fig:twoMechainForDDF}(c). Priority arbitration is performed in two stages. First, for tasks related by DDF, execution precedence is determined according to $\text{NP}[1]$, and the subset of tasks with the highest precedence, denoted by $HPS$, is selected. Since such precedence is defined only pairwise among communicating tasks and thus forms a partial order, $HPS$ may contain multiple tasks. Second, among the tasks in $HPS$, the one with the highest priority according to $\text{NP}[0]$ is selected for execution.

To mitigate the overhead of instantiating each job as a separate task, an equivalent multi-frame task model (MFTM)~\cite{baruah1999generalized, mok2002multiframe} can be used to represent variations of job attributes in a single task during hyperperiod, thus avoiding per-job task expansion.

\subsubsection{Multi-buffer Communication for Read-From-Intended (RFI)} 
\label{label:Multi-buffer Communication for Read-From-Intended}
\ 
\newline
\indent
The execution dependencies ensure the data required by a reader has been produced, but they do not guarantee that a reader can read the data from the intended writer, i.e., RFI. 
% As shown in Fig.~\ref{fig:twoMechainForDDF}(c), under register communication, job $J_6^3$  can only read data produced by $J_4^4$ due to $J_4^4$'s early finish, although (intended) $J_4^3$'s data has been produced, but overwritten by the former.
% As shown by the dashed circle in Fig.~\ref{fig:twoMechainForDDF}(c), under register communication, $J_6^3$ can only read the data produced by $J_4^4$ because $J_4^4$ finishes earlier. Although $J_6^3$ has an execution dependency on its intended writer $J_4^3$, and the data of $J_4^3$ has already been produced, it is overwritten by the later write of $J_4^4$.
% As shown by the dashed circle in Fig.~\ref{fig:twoMechainForDDF}(c), under register communication, $J_6^3$ reads the data produced by $J_4^4$ because $J_4^4$ finishes earlier. Although $J_6^3$ depended on and was intended for writer $J_4^3$, and $J_4^3$’s data has been produced, it is overwritten by the subsequent write of $J_4^4$.
As shown in Fig.~\ref{fig:twoMechainForDDF}(c), under register communication, the intended dependency (dashed arrow) $J_6^3$ \textit{after} $J_4^3$ does not guarantee that $J_6^3$ \textit{reads from} $J_4^3$, since an intervening write by $J_4^4$ makes (solid arrow) $J_6^3$ \textit{read} $J_4^4$ instead.

\begin{figure}
    \centering
    \includegraphics[width=\linewidth]{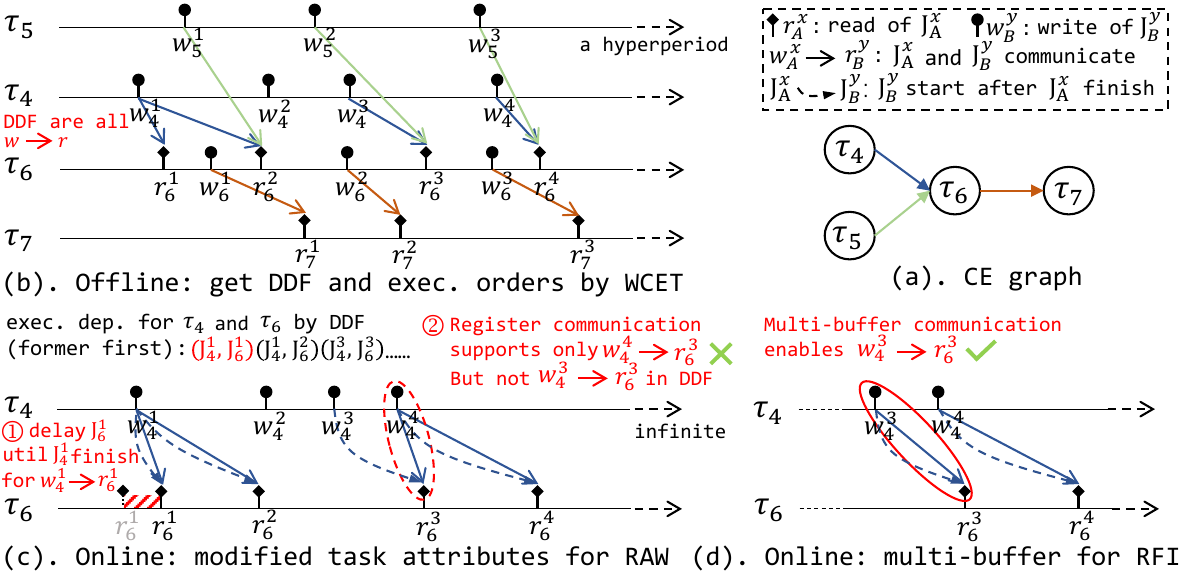}
    \vspace{-0.6cm}
    \caption{Mechanism of the RAW and RFI for the DDF
    }
    \label{fig:twoMechainForDDF}
    \vspace{-0.35cm}
\end{figure}

% of task $\tau_6$ can only read data produced by $J_4^4$ of task $\tau_4$, even though the data from $J_4^3$ has already been produced but cannot be correctly accessed.

To address this issue, we introduce multi-buffer communications in Fig.~\ref{fig:twoMechainForDDF}(d) to support storing multiple data from different writing jobs of the same task and select the intended one to ensure RFI. During online, each reading job selects and reads from the buffer corresponding to its intended writer in the DDF, thereby combining with RAW, the communication relations can be maintained online.

% Furthermore, introducing the multi-buffer mechanism requires determining the buffer size offline to ensure that sufficient historical data versions can be preserved online. 
Furthermore, introducing the multi-buffer mechanism requires determining the buffer size offline to ensure that multiple data produced by jobs of the same task can be preserved online. 
Under the implementation by the MFTM, the buffer size $bs(\tau_w)$ of write task $\tau_w$ only needs to be computed at the task level, as shown in the following. $\tau_r$ denotes a reading task, $\mathbb{RS}(\tau_w)$ represents the set of $\tau_r$s communicate with $\tau_w$, and $re^{\max}(J)$ will be provided in Eqn.~\ref{eq:read_write_bounds}.
\vspace{-0.3cm}

{\footnotesize
\begin{equation*}
bs(\tau_w) = \max_{\tau_r \in \mathbb{RS}(\tau_w)}
\Biggl\{
  \max_{J_r^k, k \in \mathbb{N}^*}
  \bigl|\{\, J_w^l \mid  we(R_{com}(J_r^k)) \le we(J_w^l) \le re^{\max}(J_r^k) \,\}\bigr )|
\Biggr\}
\end{equation*}
}

\section{Proof of the Timing Anomalies Free}
\label{sec:five}
% This chapter takes the Maximum Reaction Time (MRT) as an example to illustrate the latency analysis and prove TA-free by the invariance of immediate forward chains and safe read/write bounds.

Here, taking the MRT as an example to show the proof of TA-free under DDF the Thm.~\ref{them:tafree} by  invariance of immediate forward chains the Thm.~\ref{thm:inva} and  boundary of jobs' read/write times the Thm.~\ref{lem:boudrw}.

\subsection{Invariance of Immediate Forward Chains}

% \textbf{Theorem 1.} Under fixed communication relation, all immediate forward chains $\overrightarrow{ic}$s obtained under any feasible combination of execution times remain unchanged.
\begin {theorem}\label{thm:inva}
Under the fixed communication relations between jobs (DDF), all immediate forward chains $\overrightarrow{ic}$s obtained remain unchanged under all-WCET and any non-all-WCET executions.
% Under a fixed communication relation, all immediate forward chains $\overrightarrow{ic}$s obtained under any feasible combination of execution times remain unchanged.
\end{theorem}
\vspace{-0.1cm}
% To prove Thm.~\ref{thm:inva} we fist prove its special case Lem.~\ref{lem:siginv}
To prove the  Thm.~\ref{thm:inva}, we first prove the Lem.~\ref{lem:siginv}, a special case with a single edge in a CE chain, like the $\tau_5 \rightarrow \tau_6$ in Fig.~\ref{fig:twoMechainForDDF} (a).
% \textbf{Lemma 1.} Under fixed communication relations, all $\overrightarrow{ic}$s of any edge $\tau_m \rightarrow \tau_{m+1}$ in a CE chain remain unchanged.
\vspace{-0.1cm}
\begin {lemma}\label{lem:siginv} Under DDF, $\overrightarrow{ic}$s of any edge $\tau_m \rightarrow \tau_{m+1}$ in a CE chain are unchanged under all-WCET and any non-all-WCET executions.
\end{lemma}
% \begin {lemma}\label{lem:siginv} Under fixed communication relations, all $\overrightarrow{ic}$s of any edge $\tau_m \rightarrow \tau_{m+1}$ in a CE chain remain unchanged.
% \end{lemma}
% \textbf{Lemma 1.} Under fixed communication relations, all $\overrightarrow{ic}$s of any edge $\tau_m \rightarrow \tau_{m+1}$ in a CE chain remain unchanged.

\vspace{-0.3cm}
\begin{proof}
% According to Definition 1, for an edge $\tau_m \rightarrow \tau_{m+1}$, the construction process of the $\overrightarrow{ic}$ requires selecting, for each write job $J^r_m$ of task $\tau_m$, a corresponding read job  $J^{r'}_{m+1}$ from the job set of task $\tau_{m+1}$ as its direct successor $J_{m+1}$.
% If, for every $J^r_m$, the corresponding $J^{r'}_{m+1}$ is unchanged under all combinations of jobs' execution times, then the $\overrightarrow{ic}$ of $\tau_m \rightarrow \tau_n$ not change. As stated in Definition 1, there are two cases in determining the successor job $J_{m+1}$ for $J_{m}$:
According to Definition 1, for an edge $\tau_m \rightarrow \tau_{m+1}$, the construction process of the $\overrightarrow{ic}$ requires selecting, for each write job $J^r_m$ of task $\tau_m$, a corresponding read job  $J^{r'}_{m+1}$ from the job set of task $\tau_{m+1}$ as its direct successor $J_{m+1}$.
If, for every $J^r_m$, the corresponding $J^{r'}_{m+1}$ is unchanged under all-wcet and any non-all-wcet executions, then the $\overrightarrow{ic}$ of $\tau_m \rightarrow \tau_n$ not change. As stated in Definition 1, there are two cases in determining the successor job $J_{m+1}$ for $J_{m}$:

Case 1: If $J_m \in wj(R_{\text{com}},m+1)$, i.e., $J_m$ is writing with job of $\tau_{m+1}$ in DDF, its successor $J_{m+1}$ is directly the reader that communicates with $J_m$.
Since the communication is fixed, $J_{m+1}$ is unchanged.

Case 2: If $J_m \notin wj(R_{\text{com}}, m+1)$, then $J_{m+1}$ corresponds to the reading job $J^{l'}_{m+1}$ of task $\tau_{m+1}$ associated with the first subsequent job $J^{r*}_m$ of the same task $\tau_{m}$ for $J_m$ that satisfies Case 1.
Similarly, the communication relation is fixed, both $J^{r*}_m$ and $J^{l'}_{m+1}$ are uniquely unchanged. So $J_{m+1}$ also remains unchanged.

% According to Definition 1, for an edge $\tau_m \rightarrow \tau_{m+1}$, the construction process of the $\overrightarrow{ic}$ requires selecting, for each job $j^m_k$ of task $\tau_m$, a corresponding job $j^{m+1}_l$ from the job set of task $\tau_{m+1}$ as its direct successor $J_{i+1}$.
% If, for every $j^m_k$, the corresponding $j^{m+1}_l$ remains unchanged under all combinations of execution times, then the $\overrightarrow{ic}$ corresponding to $\tau_m \rightarrow \tau_n$ not change. As stated in Definition 1, there are two cases in determining the successor job $J_{i+1}$:

% If $J_i \in wj(R_{\text{com}})$, i.e., $J_i$ is writing task in DDF, its successor $J_{i+1}$ is directly the reading job that communicates with it.
% Since the communication relation is fixed, $J_{i+1}$ is unchanged.

% If $J_i \notin wj(R_{\text{com}})$, then $J_{i+1}$ corresponds to the reading job $j^{m+1}_{l'}$ associated with the first subsequent job $j^m_{k'}$ of the same task that satisfies Case 1.
% Similarly, the communication relation is fixed, both $j^m_{k'}$ and $j^{m+1}_{l'}$ are uniquely determined. So $J_{i+1}$ also remains unchanged.

Therefore, Lem.~\ref{lem:siginv} is proven.
\vspace{-0.1cm}
\end{proof}
\vspace{-0.1cm}

We now prove that the theorem holds for any complete CE chain %$\tau_{1} \rightarrow \tau_{2} \rightarrow \dots \rightarrow \tau_{n}$ 
$\tau_{1} \rightarrow \dots \rightarrow \tau_{n}$
with a length greater than 1, i.e.,  proving Thm~\ref{thm:inva}
\vspace{-0.2cm}
\begin{proof}
Base Case:
For $n=2$, Lem.~\ref{lem:siginv} proves that $\overrightarrow{ic}$ of $\tau_1 \rightarrow \tau_2$ remains unchanged under fixed communication relations (DDF).

Induction Hypothesis:
Assume that the proposition holds for $n = k$; that is, when the communication relations are fixed, all $\overrightarrow{ic}$s of the
$E_q = \tau_1 \rightarrow \tau_2 \rightarrow \dots \rightarrow \tau_k$
remain unchanged.

Inductive Step:
For a CE chain $E_l$ with length $k + 1$,
for Lem~\ref{lem:siginv}, the $\overrightarrow{ic}$ of the edge $\tau_k \rightarrow \tau_{k+1}$ remains unchanged when the communication relation is fixed.
That is, for any writing job $J_k^r$ of the writing task $\tau_k$, there exists a fixed reading job $J_{k+1}^{r'}$ from $\tau_{k+1}$.
Therefore, once the $\overrightarrow{ic}$s of $E_q$ are unchanged, extending them to $E_l$ not alter the successor-selection and thus the  $\overrightarrow{ic}$ of $E_l$ is unchanged.

Therefore, Thm~\ref{thm:inva} is proven.
\end{proof}

% \vspace{-0.5cm}
% \subsection{Bound Times of Job's Read and Write Event}
\subsection{Boundary of Jobs' Read and Write Times}
% \textbf{Lemma 2.}
% For a periodically triggered system scheduled by FP scheduling, if all jobs execute with WCETs, resulting in a scheduling $S_{W}$, the start and finish time of each job represent the upper bounds of their corresponding times under all scheduling $S \in \mathbb{S}$ by feasible combinations of various in jobs' execution times. Similarly, execution based on BCET yields the corresponding lower bound\cite{gunzel2023compositional}.
\begin{theorem}\label{lem:boudrw}
% For a periodically triggered system scheduled by FP scheduling, if jobs all-WCET execution, resulting in a scheduling $S_{W}$, the start and finish time of each job represent the upper bounds of their corresponding times under all scheduling $S \in \mathbb{S}$ by feasible combinations of various in jobs' execution times. Similarly, all-BCET execution yields the corresponding lower bound\cite{gunzel2023compositional}.
For a system under periodically time-triggered FP (TT-FP) scheduling, the schedule obtained with all-WCET execution, denoted as $S_W$, provides upper bounds on the start and finish times of all jobs across any possible schedule $S \in \mathbb{S}$ from all-WCET and any non-all-WCET executions.
Similarly, all-BCET execution can yield the corresponding lower bound\cite{gunzel2023compositional}.
\end{theorem}
% \vspace{-0.1cm}
% Although the modification of task attributes proposed in Sec. \ref{sec:Modification of Task Attribution} adopts a priority group, all priorities are statically assigned to tasks in the offline and remain unchanged during online execution. And the scheduler selects the ready task with the highest priority all the time. Moreover, the system still operates under a periodic TT. Therefore, our scheduling satisfies the conditions for Thm~\ref{lem:boudrw}.

Although Sec.~\ref{sec:Modification of Task Attribution} introduces priority groups through task-attribute modification, all priorities are statically determined offline and remain unchanged online. Since the scheduler always selects the highest-priority ready task and the system still follows periodic TT scheduling, the conditions of Thm.~\ref{lem:boudrw} remain satisfied.

Under the scheduling with our treatment, by executing jobs with WCETs and BCETs, respectively in scheduling window $[0, 2H+O_{\max}]$, we can get the upper and lower bounds of the start and finish times for any job $J$, denoted as,$\ s_{S_{W}}(J)$,$\ f_{S_{W}}(J)$,$\ s_{S_{B}}(J)$,$\ f_{S_{B}}(J)$.

Since adopting implicit communications, the upper and lower bounds of its read event and write event time in all schedules $\mathbb{S}$ are equivalent to the bounds of start and finish times, which are:
% \vspace{-0.1cm}
{\small
\begin{equation}
\begin{aligned}
re^{\min}(J) &= s_{S_{B}(J)} \quad &re^{\max}(J) &= s_{S_{W}(J)} \\
we^{\min}(J) &= f_{S_{B}(J)} \quad &we^{\max}(J) &= f_{S_{W}(J)}
\end{aligned}
\label{eq:read_write_bounds}
\end{equation}
}

\subsection{Timing Anomalies Free for MRT}
% From Chapter II, the immediate forward augmented chain ($\overrightarrow{iac}$) can be derived from the $\overrightarrow{ic}$. Specifically, the write event time of the last job in $\overrightarrow{ic}$, denoted as $we(\overrightarrow{ic}[-1])$, is appended to the end of $\overrightarrow{ic}$ as $z'$, while the read event time of the preceding job $pre(\overrightarrow{ic}[0])$—which belongs to the same task as the first job $\overrightarrow{ic}[0]$—is inserted at the beginning of $\overrightarrow{ic}$ as $z$.

% According to Theorem 1, under fixed communication relations, $\overrightarrow{ic}$ remains invariant for any feasible combination of execution times for jobs. Therefore, the upper bound of $\overrightarrow{iac}$ is determined solely by the minimum of $re(pre(\overrightarrow{ic}[0]))$ and the maximum of $we(\overrightarrow{ic}[-1])$, both of which can be obtained from Eqn. (3). Consequently, under schedule $S$ and fixed communication relations, the upper bound of $\overrightarrow{iac}$ for $E$ over all combinations of execution times is Eqn (4).
% \begin{equation}
% \ell^{max}(\overrightarrow{iac}_{E,S}) = we^{max}(\overrightarrow{ic} [ -1 ]) - re^{min}(pre(\overrightarrow{ic} [ 0]))
% \end{equation}

% The formulation of MRT is adapted from Eqn. (2) in Chapter II, as given by the following equation.
% \vspace{-0.1cm}
% \begin{equation*}
% \mathrm{MRT}(E,S) := \max \Bigl\{ \ell^{max}\bigl(\overrightarrow{iac}^m_{E,S}\bigr) \ \big|\ m \in \mathbb{N},\ \overrightarrow{c}_{E,S}^m\ \text{valid} \Bigr\}.
% \end{equation*}
Eqn.~\ref{label:CompMRT1} computes the MRT for a specific schedule $S$ obtained under a fixed scheduling algorithm, corresponding to one specific combination of job execution times. We now extend it to encompass all schedules $\mathbb{S}$ under the same fixed scheduling algorithm, as follows:
% Eqn.~\ref{label:CompMRT1} computes the MRT for a CE chain $E$ under a fixed scheduling algorithm, given any feasible combination of job execution times corresponding to a schedule $S$. Extending this definition, the MRT of $E$ over all schedules $\mathbb{S}$ is:
{\small
\begin{equation}
    MRT(E) =  \max \left\{ \max \Bigl\{ \ell\bigl(\overrightarrow{iac}_{E,S}^m\bigr) \ \big|\ m \in \mathbb{N},\ \overrightarrow{iac}_{E,S}^m\ \text{valid} \Bigr\} \ \big| \ S \in \mathbb{S} \right\}
    \label{label:CompMRT2} 
\end{equation}
}

\begin{theorem}[TA Free]\label{them:tafree}
 When the system executes under the DDF, the MRT obtained from the schedule $S_W$ by all-WCET execution offline, i.e., $MRT(E, S_{W})$, is not exceeded by the results from the schedules $\mathbb{S}$
 by all-WCET and any non-all-WCET executions online, i.e.,
 %  When the system executes under DDF, the MRT obtained from the offline by all-WCET execution, i.e., $MRT(E, S_{W})$, is not exceeded by any scheduling 
 % from all feasible combinations of various jobs' execution times
 % in online execution, i.e.,
\end{theorem}
\vspace{-0.4cm}
% \textbf{Theorem 2. (TA-free)} 
%  When the system executes under DDF, the MRT obtained from the offline schedule by all WCETs $MRT(E, S_{W})$, is not exceeded by any scheduling 
%  from all feasible combinations of various jobs' execution times
%  in online execution, i.e.,
{\small
 \begin{equation*}
     MRT(E) = MRT(E,S_{W}) = \max \Bigl\{ \ell\bigl(\overrightarrow{iac}_{E,S_{W}}^m\bigr) \ \big|\ m \in \mathbb{N},\ \overrightarrow{iac}_{E,S_{W}}^m\ \text{valid} \Bigr\}.
 \end{equation*}
}
\vspace{-0.5cm}
\begin{proof}
By Thm~\ref{thm:inva}, under fixed communication relations (DDF), $\overrightarrow{ic}_{E,S}^m$ of sampling point $m$ remains unchanged across any non-all-WCET exections compared to all-WCET, i.e., $\overrightarrow{ic}_{E,S}^m = \overrightarrow{ic}_{E,S_{W}}^m$. Therefore, no new $\overrightarrow{ic}$s are generated during online execution.

According to Sec.~\ref{sec:Job Chain and End-to-End latency}, the $\overrightarrow{iac}_{E,S}^m$ can be derived from $\overrightarrow{ic}_{E,S}^m$. The write event time $we(\overrightarrow{ic}_{E,S}^m[-1])$ of the last job in $\overrightarrow{ic}_{E,S}^m$ is appended as $z'$ to the end of the chain, while the read time $re(pre(\overrightarrow{ic}_{E,S}^m[0]))$ of the previous job of the first job’s task $pre(\overrightarrow{ic}_{E,S}^m$ $[0])$ is added as $z$ to the start of the chain. Since the read event time (sampling) of the first task in the CE chain coincides with its release time as defined in the last of Sec.~\ref{lab:System Model}, we have $re(pre(\overrightarrow{ic}_{E,S}^m[0])) = r(pre(\overrightarrow{ic}_{E,S}^m[0]))$, which is a constant can be computed by its period, phase and id.

Furthermore, since $\overrightarrow{ic}_{E,S}^m$ is identical to $\overrightarrow{ic}_{E,S_{W}}^m$ for all $S \in \mathbb{S}$. But variation lies in the job finishing times for different $S$, leading to multiple $\overrightarrow{iac}_{E, S}^m$ corresponding to a $\overrightarrow{ic}_{E, S_W}^m$. The max length of them is: $max \{ \ell(\overrightarrow{iac}_{E,S}^m) = f_{S}(\overrightarrow{iac}_{E,S}^m [ -1 ]) - r(pre(\overrightarrow{iac}_{E,S}^m [ 0])) \ | \ S \in \mathbb{S} \} $.

% This expression reaches its maximum when $f_S(\overrightarrow{iac}_{E,S}^m[-1])$ is max. According to Eqn.~\ref{eq:read_write_bounds} from Thm~\ref{lem:boudrw}, this maximum occurs under the WCET-based schedule $S_{W}$, i.e., $\max \{ \ell(\overrightarrow{iac}_{E,S}^m) \ | \ S \in \mathbb{S} \} = \ell(\overrightarrow{iac}_{E,S_{W}}^m)$.
%  Therefore, it follows that: $MRT(E) = MRT(E,S_{W})$. 
This expression reaches its maximum when $f_S(\overrightarrow{iac}_{E,S}^m[-1])$ is max. According to Eqn.~\ref{eq:read_write_bounds} from Thm~\ref{lem:boudrw}, this maximum occurs when $S$ is the all-WCET schedule $S_{W}$. So, $\max \{ \ell(\overrightarrow{iac}_{E,S}^m) \ | \ S \in \mathbb{S} \} = \ell(\overrightarrow{iac}_{E,S_{W}}^m)$.
 Therefore, it follows that: $MRT(E) = MRT(E,S_{W})$. 
 % Theorem 3 is proven, TA eliminated for MRT, and the same for MDA omitted.
 % Theorem 3 is thus proven, implying that TAs are eliminated for MRT; the corresponding proof for MDA is omitted for brevity. 
 Thm~\ref{them:tafree} is thus proven, implying that TA-free for MRT.
 \end{proof}
 \vspace{-0.2cm}
The same method can be applied and proved TA-free for MDA.

\section{Experiment Evaluation}
\begin{figure}[!t]
    \centering
    \includegraphics[width=0.95\linewidth]
    {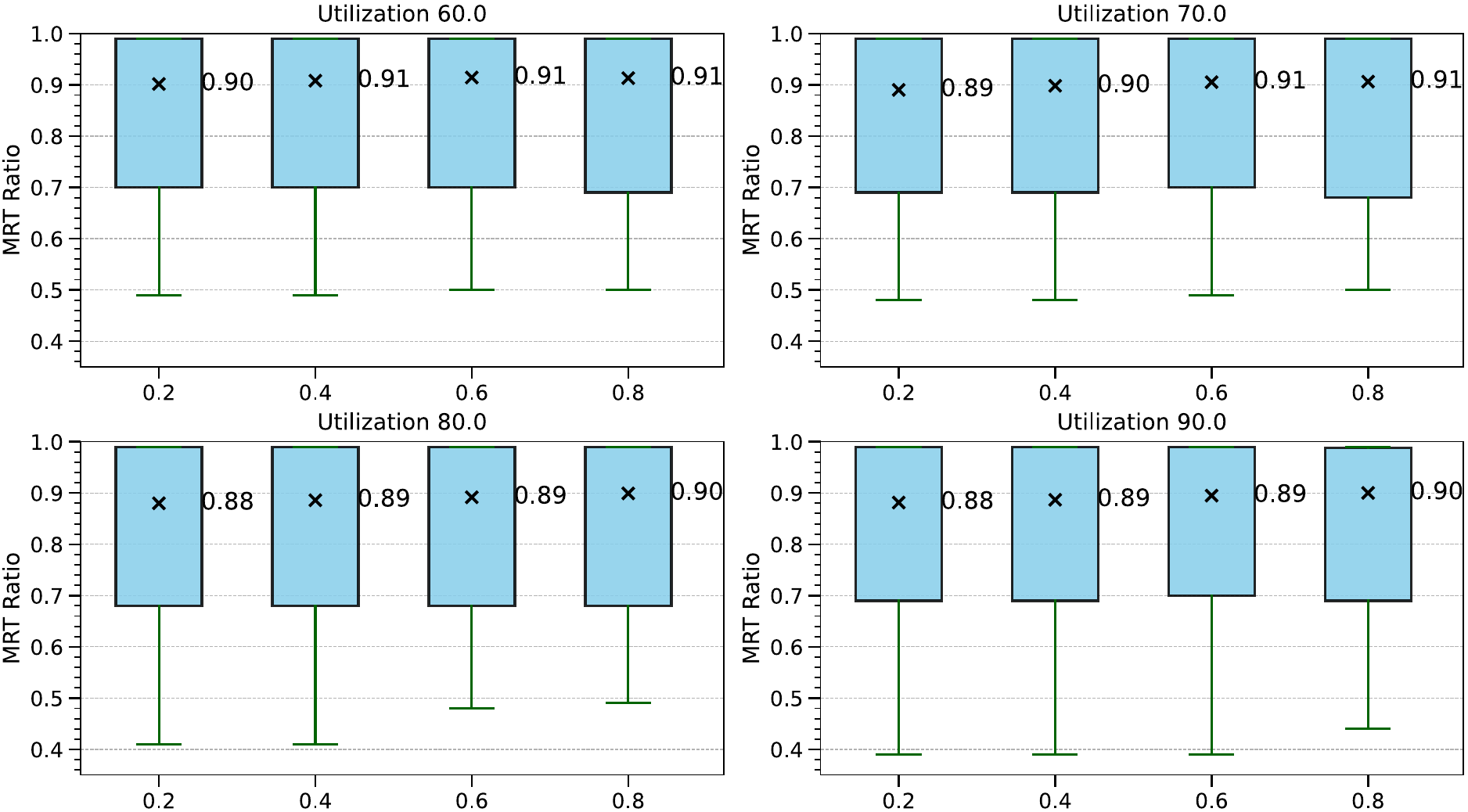}
    \vspace{-0.3cm}
    \caption{ MRT ratio for Our/M23 in various $\alpha$ and $U$
    }
    \vspace{-0.1cm}
    \label{fig:MRTRtio}
    \vspace{-0.4cm}
\end{figure}

% \begin{figure*}[!t]
%     \centering
%     % === 子图 (a) ===
%     \subfloat[MRT theoretical jitter]{%
%         \includegraphics[width=0.32\textwidth]{picture/exper/E3_5_test.png}
%     }\hfill
%     % === 子图 (b) ===
%     \subfloat[average reaction time on online]{%
%         \includegraphics[width=0.32\textwidth]{picture/exper/E2_6_test.png}
%     }\hfill
%     % === 子图 (c) ===
%     \subfloat[multi-buffer's memory overhead]{%
%         \includegraphics[width=0.32\textwidth]{picture/exper/E5_6_test.png}
%     }
%     \vspace{-0.2cm}
%     \caption{ Results for three experiments}
%     \vspace{-0.4cm}
%     \label{fig:three_in_a_row}
% \end{figure*}

\begin{figure*}[!t]
    \centering
    \begin{minipage}[t]{0.31\textwidth}
        \centering
        \includegraphics[width=\textwidth]{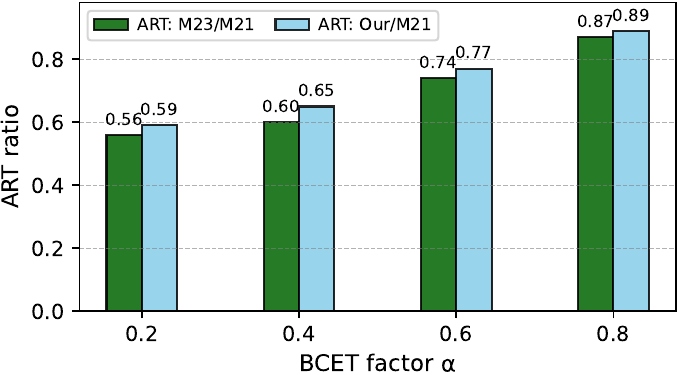}
        \vspace{-0.7cm}
        \caption{Online average reaction time}
        \vspace{-0.4cm}
        \label{fig:avg_rt}
    \end{minipage}\hfill    
    \begin{minipage}[t]{0.31\textwidth}
        \centering
        \includegraphics[width=\textwidth]{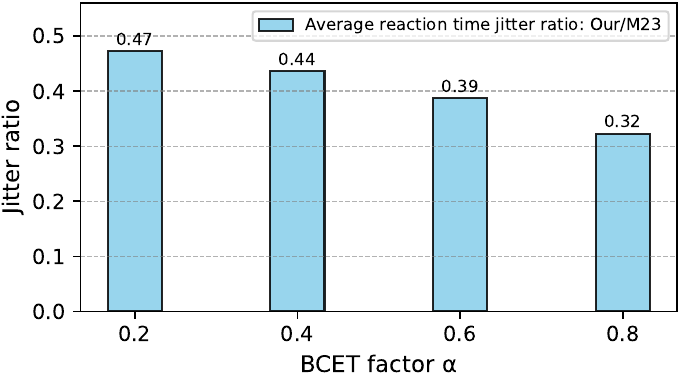}
        \vspace{-0.7cm}
        \caption{MRT jitter}
        \vspace{-0.4cm}
        \label{fig:mrt_jitter}
    \end{minipage}\hfill
    \begin{minipage}[t]{0.332\textwidth}
        \centering
        \includegraphics[width=\textwidth]{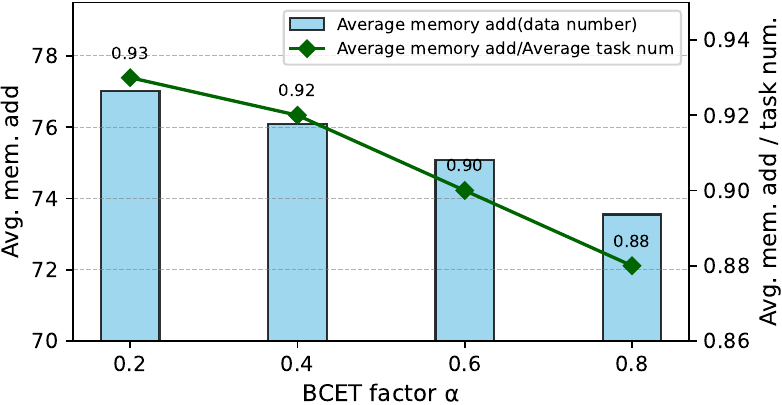}
        \vspace{-0.7cm}
        \caption{Memory overhead}
        \vspace{-0.4cm}
        \label{fig:mem_overhead}
    \end{minipage}
    
\end{figure*}

The end-to-end latency is critical for verifying the timing behavior of CE chains in automotive systems.
We evaluated our treatment using synthetic task sets and chains derived from the Automotive Benchmarks~\cite{kramer2015real}, and compared it with SOTA methods~\cite{gunzel2021timing, gunzel2023compositional}.
% Evaluation metrics include the MRT, average reaction times (ART), and jitters of reaction time, with an analysis of schedulability. Furthermore, we analyzed the memory overhead introduced by the multi-buffer communication.

\subsection{Task Set and Cause-effect Chain Generation}
Task sets and CE-chains generation by the method in~\cite{gunzel2023compositional} with the configurations in Automotive Benchmarks~\cite{kramer2015real}. 

Each task $\tau_i$ is characterized by its WCET $W(\tau_i)$, period $T(\tau_i)$, phase $O(\tau_i)$, and priority $P(\tau_i)$.
% The task generation process adheres to the~\cite{kramer2015real} and proceeds as follows: 
The task generation process is outlined as follows:
$T(\tau_i)$ is randomly selected from $\{1, 2, 5, 10, 20, 50,$ $ 100, 200, 1000\}$ (ms) based on~\cite{kramer2015real}; $W(\tau_i)$ is get by multiplying its average execution time—generated according to a Weibull distribution satisfying constraints in~\cite{kramer2015real}, with a random scaling factor uniformly distributed in $[f_{\min}, f_{\max}]$; Before applying our treatment, the system uses RM scheduling, so $P(\tau_i)$ is determined by $T(\tau_i)$. Our treatment is also applicable to other FP scheduling.

Task utilization is $U_i = C_i / T_i$.
For target total utilizations $U \in \{60\%, 70\%, 80\%, 90\% \}$, 1000 task sets are generated.
Each task set is created by randomly generating 1000–1500 candidate tasks, then selecting a subset $\mathcal{T}$ of candidate tasks to satisfy $|\sum_{\mathcal{T}} U_i - U| \le 0.01$.

Based on the generated task sets, a corresponding set of CE chains is constructed according to the rules in~\cite{kramer2015real}. Specifically, for each CE chain, the number of activation modes $P_j \in \{1,2,3\}$ is randomly selected following the distribution in Table VI of~\cite{kramer2015real}. Then, $P_j$ distinct periods are randomly chosen from the task set to form the period set $T_j$. For each period in $T_j$, 2--5 tasks are further selected without replacement from the corresponding task group according to the distribution in Table VII of~\cite{kramer2015real}. Cyclic dependencies are explicitly excluded to ensure that the resulting CE graph remains acyclic. Each task set finally yields 30--60 CE chains, forming the benchmark used for subsequent evaluation.

\subsection{Evaluation Results}
% For each utilization $U \in \{60\%,70\%,80\%,90\%\}$, we generate 1000 benchmarks. 
% For each benchmark and BCET factor $\alpha \in \{0.2,0.4,0.6,0.8\}$ (with $B(\tau_i)=\alpha W(\tau_i)$), we evaluate all performance metrics under the three methods.

% For each utilization level $U \in \{60\%, 70\%, 80\%, 90\%\}$, we generate 1000 benchmarks. 
% For every benchmark, we conduct four groups of experiments corresponding to BCET factors $\alpha \in \{0.2, 0.4, 0.6, 0.8\}$, which specifing each task's BCET $B(\tau_i)=\alpha \cdot W(\tau_i)$.
% Each group evaluates multiple performance metrics under the three methods.

We conducted experiments with the 1000 task sets and their chains for each combinations of BCET factors $\alpha \in \{0.2, 0.4, 0.6, 0.8\}$ and $U \in \{60\%,70\%,80\%,$ $90\%\}$, where $\alpha$ sets task’s BCET to $\alpha \cdot W(\tau_i)$.

% $BCET = W(\tau_i) \times \alpha$.

% We conducted experiments with the 1000 task sets for different combinations of BCET factors $\alpha = \{0.2, 0.4, 0.6, 0.8\}$ and the aforementioned values of $U \in \{60\%,70\%,80\%,90\%\}$, where $BCET = W(\tau_i) \times \alpha$.

\subsubsection{Max reaction time (MRT)}
We compared the MRTs obtained by RM scheduling using a safe analysis method for TA in~\cite{gunzel2023compositional} (M23) and TA-free treatment by enforcing all-WCET execution online ~\cite{gunzel2021timing} (M21) with those after applying our treatment across different $U$ and $\alpha$. 
The $\frac{MRT_{Our}}{MRT_{M23}}$ and $\frac{MRT_{Our}}{MRT_{M21}}$ was computed for each CE chain.

% In Fig.~\ref{fig:MRTRtio}, the boxplot indicates that $MRT_{OUR}$ is smaller than $MRT_{M23}$ in about 80\%–96\% of cases, with a maximum reduction of up to 60\%. This improvement stems from fixing the communication relations, which prevents changes in the structure of immediate forward chains ($\overrightarrow{ic}$) and thus eliminates periodic-level delay accumulation by job-selection shifts.

The boxplot in Fig.~\ref{fig:MRTRtio} illustrates the ratio $\frac{MRT_{Our}}{MRT_{M23}}$ for all CE chains exhibiting TAs. Chains without anomalies are not shown, as their results are identical. It can be observed that $MRT_{Our}$ is consistently lower than $MRT_{M23}$, with an average reduction of 9\%–12\%, and up to 61\% in the best case. This demonstrates that our method effectively reduces the MRT. This improvement stems from fixing communication relations, which prevents changes in the structure of immediate forward chains and thus eliminates periodic-level delay accumulation from job-selection shifts. 
%Since $MRT_{our}$ is based on the actual execution with WCET, it serves as an exact upper bound. In contrast, $MRT_{M23}$, computed based on 抽象解释去构造可能是一个近似的上界.
% Since $MRT_{our}$ is based on the actual execution, it serves as an exact upper bound. In contrast, $MRT_{M23}$, computed based on the construction of abstract interpertain, may provide an approximate upper bound.

$MRT_{our}$ is an exact bound based on a actual execution. But $MRT_{M23}$ is approximate, constructed by abstract interpretation~\cite{cousot1977abstract}.

% In a few cases, $MRT_{OUR}$ is slightly higher than $MRT_{M23}$ (no more than 5\%). This occurs when the execution dependencies postpone the release of certain tasks, indirectly advancing others’ start times. If the chain causing the MRT coincides with the one derived under WCET, the earlier start slightly increases the reaction time, though the effect remains marginal.

In addition, since $M21$ eliminates TAs, their MRT equals ours and is not shown in Fig.~\ref{fig:MRTRtio}. However, $M21$ achieves TA-free by enforcing all-WCET execution online, causing large average reaction times.

\subsubsection{Average reaction time (ART)}
\label{label:Average reaction time}
For each system, we performed online simulation executions where the execution times of all jobs in $M23$ and ours were randomly sampled from the $[BCET, WCET]$. The ART ratio of each CE chain was computed $\frac{ART_{Our}}{ART_{M21}}$ and $\frac{ART_{M23}}{ART_{M21}}$. Fig.~\ref{fig:avg_rt} shows that both $M23$ and our method achieve significantly lower ARTs than $M21$, as $M21$ eliminates TAs by enforcing all-WCET execution online ($ART=MRT$).
Compared with M23, which retains the original scheduling and thus achieves the lowest average latency but still has TAs, our approach incurs only a 2\%–5\% increase.

\subsubsection{Reaction time jitter}
We computed the jitter ratio between $M23$ and ours as $\frac{MRT_{Our} - mRT_{Our}}{MRT_{M23} - mRT_{M23}}$, and averaged the results across $\alpha$. The min reaction time ($mRT$) of our can be derived by the Eqn.~\ref{eqn:mRT}, whose proof, like MRT, is omitted. $mRT_{M23}$ is obtained by extending the method in~\cite{gunzel2023compositional} to construct the shortest abstract chain.

As shown in Fig.~\ref{fig:mrt_jitter}, our method consistently has a lower average jitter than $M23$. This is because fixing the structure of the immediate forward chain reduces timing uncertainty due to job dependency variations, yielding a more stable MRT and lower jitter. 
% \vspace{-0.0cm}
{\footnotesize
 \begin{equation}
     mRT_{Our}(E) = mRT(E,S_{B}) = \min \Bigl\{ \ell\bigl(\overrightarrow{iac}_{E,S_{B}}^m\bigr) \ \big|\ m \in \mathbb{N},\ \overrightarrow{iac}_{E,S_{B}}^m\ \text{valid} \Bigr\}
     \label{eqn:mRT}
 \end{equation}
}

The method in $M21$, enforcing all-WCET executions online, has a zero jitter, but a pessimistic ART in Sec. ~\ref{label:Average reaction time}.

\subsubsection{Memory overhead of the multi-buffer communication}
As shown in Fig.~\ref{fig:mem_overhead}, the multi-buffer communication only increases slightly in overall memory usage (measured by the number of stored data items) computed by the method in Sec.~\ref {label:Multi-buffer Communication for Read-From-Intended}.

\subsubsection{Schedulability}
% Our treatment delays execution of certain tasks by adjusting their release times, thereby narrowing their feasible scheduling windows. However, our analysis across thousands of task sets for different utilization shows that if a task set is schedulable under the original scheduling, it remains schedulable after applying our treatment. 

% Our treatment delays execution of certain tasks by adjusting their release times, thereby narrowing their feasible scheduling windows. Although this could theoretically reduce schedulability, analysis across thousands of task sets for different utilization shows that if a task set is schedulable under the original scheduling, it remains schedulable after applying our treatment. Thus, our treatment has a negligible impact on schedulability.

Our treatment delays certain tasks by adjusting their release times, narrowing their scheduling windows. However, analysis of thousands of task sets with varying utilization shows that if a task set is schedulable under the original scheduling, it remains schedulable after applying our treatment.

% \vspace{-0.2cm}
\section{RELATED WORK}
The problem of timing anomalies (TAs)~\cite{reineke2006definition, lundqvist1999timing, 9904748, binder2021still, cassez2012timing,HTA} has been extensively studied since its introduction by Graham (1966)~\cite{graham1969bounds}. Most existing work focuses on WCRT analysis. Voudouris et al.~\cite{voudouris2017timing} proposed LAZY, a TA-free scheduling algorithm for event-triggered non-preemptive multicore systems. Chen et al.~\cite{TAFree1} extended it to condition DAG models, and Dai et al.~\cite{dai2021timing} further applied it to periodic non-preemptive tasks. Lin et al.~\cite{lin2023scheduling} systematically analyzed TAs in self-suspending models, and more recently,  Lin et al.~\cite{lin2025quasi} explored TAs in scheduling the Lingua Franca programs.

However, research~\cite{gunzel2021timing,gunzel2023compositional, TA_analysis_1} on TAs in end-to-end latency~\cite{davare2007period, vincentelli2007optimizing, zhu2013optimization, hamann2017communication, teper2022end,gunzel2021timing,gunzel2023compositional, TA_analysis_1} remains limited. Compared to the WCRT, end-to-end anomalies are more complex, as they depend on individual job response variations and inter-job interactions in a chain.  G{\"u}nzel et al.~\cite{gunzel2021timing} addressed TAs at the scheduling level by enforcing WCET execution, which, however, degrades average performance. Other works~\cite{gunzel2023compositional, TA_analysis_1} proposed safe analysis methods that bound latency without modifying scheduling, but since TAs still exist, the latency remains large.

\vspace{-0.2cm}
\section{CONCLUSION}
% This work proposed the first treatment to eliminate TA in end-to-end latency for cause–effect chains with negligible loss in average latency, which is based on Deterministic Data Flow (DDF). Experimental results show that our method not only effectively removes timing anomalies but also significantly reduces the max end-to-end latency, average end-to-end latency, and end-to-end latency jitter.

% This work proposed the first treatment to eliminate timing anomalies (TAs) in end-to-end latency for cause–effect chains with negligible loss in average latency, which is based on Deterministic Data Flow. 
% In addition, we formally prove its TA-free property.
% Experimental results show our method significantly reduces the maximum end-to-end latency, average end-to-end latency, and latency jitter.
This work proposed the first treatment to eliminate timing anomalies (TAs) in end-to-end latency for cause–effect chains, with negligible average latency loss, based on Deterministic Data Flow. 
% Experimental results show our method significantly reduces the maximum end-to-end latency, average end-to-end latency, and latency jitter.

% In this work, we presented the first approach to eliminating timing anomalies (TAs) in the end-to-end latency of cause--effect chains with negligible average-latency loss, based on Deterministic Data Flow. We formally established its TA-free property and demonstrated through experiments that it significantly reduces maximum and average end-to-end latency, as well as latency jitter.

% \begin{acks}
% This work was supported in part by the Strategic Priority Research Program of the CAS under Grant Nos. XDB0660101, XDB0660000, and XDB0660100, in part by the NSFC under Grant No. 62502489.
% % , and in part by the Jiangsu Provincial Natural Science Foundation under Grant Nos. BK20241818, BK20251815, and BK20250479.
% \end{acks}

\vspace{0.2cm}
\noindent
\textbf{Acknowledgments}
\small
This work was supported in part by the Strategic Priority Research Program of the Chinese Academy of Sciences, Grant Nos. XDB0660101, XDB0660000, and XDB0660100, in part by the National Natural Science Foundation of China under Grant and 62502489, in part by Jiangsu Provincial Natural Science Foundation under Grants BK20241818, BK20251815, and BK20250479,  in part by  Anhui Provincial Natural Science Foundation under Grant 2508085MF141.

% \begin{acks}
% \end{acks}


\let\balance\relax 
\bibliography{reference}

\appendix

\end{document}